\documentclass[a4paper, amsfonts, amssymb, amsmath, reprint, showkeys, nofootinbib, twoside,notitlepage,onecolumn]{revtex4-1}

\def\prd{{Phys.~Rev.~D}}        
\usepackage{amsmath,amstext}
\usepackage[T1]{fontenc}
\usepackage{amssymb}
\usepackage{graphicx}
\usepackage{ae,aecompl}

\DeclareFontFamily{OT1}{pzc}{}
\DeclareFontShape{OT1}{pzc}{m}{it}{<-> s * [1.10] pzcmi7t}{}
\DeclareMathAlphabet{\mathpzc}{OT1}{pzc}{m}{it}

\usepackage{hyperref}
\usepackage{amsmath}
\usepackage{amssymb}
\usepackage{mathtools}
\usepackage{bm}
\usepackage{cleveref}
\usepackage{tensor}
\usepackage{braket}
\usepackage{enumitem}
\usepackage{mhchem}
\usepackage{amsthm}
\usepackage{nccmath}
\usepackage{mathrsfs}
\usepackage{color}

\def\be{\begin{equation}}
\def\ee{\end{equation}}
\def\beq{\begin{eqnarray}}
\def\eeq{\end{eqnarray}}

\theoremstyle{definition}

\theoremstyle{theorem}
\newtheorem{theorem}{Theorem}

\theoremstyle{corollary}

\begin{document}
\title{The diffusion equation is compatible with special relativity}
\author{L.~Gavassino}
\affiliation{Department of Applied Mathematics and Theoretical Physics, University of Cambridge, Wilberforce Road, Cambridge CB3 0WA, United Kingdom}

\begin{abstract}
Due to its parabolic character, the diffusion equation exhibits instantaneous spatial spreading, and becomes unstable when Lorentz-boosted. According to the conventional interpretation, these features reflect a fundamental incompatibility with special relativity. In this Letter, we show that this interpretation is incorrect by demonstrating that any smooth and sufficiently localized solution of the diffusion equation is the particle density of an exact solution of the relativistic Vlasov-Fokker-Planck equation. This establishes the existence of a causal, stable, and thermodynamically consistent relativistic kinetic theory whose hydrodynamic sector is governed \emph{exactly} by diffusion at all wavelengths. We further demonstrate that the standard arguments for instability arise from considering solutions that admit no counterpart in kinetic theory, and that apparent violations of causality disappear once signals are defined in terms of the underlying microscopic data.
\end{abstract} 
\maketitle

{\it \noindent \textbf{Introduction --}} Fick's law of diffusion \cite{Fick1855},
\vspace{-0.25cm}
\begin{equation}\label{diffusion}
\partial_t n = \mathfrak{D}\,\partial_x^2 n \, ,
\end{equation}
is among the most ubiquitous equations in physics. It describes the evolution of a conserved density \(n(t,x)\) whose microscopic carriers undergo stochastic motion characterized by a diffusivity \(\mathfrak{D}\). As a result, it provides the universal long-wavelength description of transport processes in a wide variety of systems \cite{KadanoffMartin1963}, including particle diffusion in fluids \cite{peliti_book}, heat conduction in both fluids \cite{landau6} and solids \cite{landau7}, and the probability spreading of Brownian motion \cite{Einstein1905Brownian}.

Despite its broad applicability, the diffusion equation is widely regarded in the relativistic hydrodynamics literature as a paradigmatic example of a dissipative theory incompatible with special relativity
\cite{MorseFeshbach1953,cattaneo1958,Israel_Stewart_1979,Hiscock_Insatibility_first_order,Baier2008,Muller_book,Jou_Extended,rezzolla_book,RomatschkeReview:2017ejr,GavassinoFronntiers2021,GavassinoSuperlum2021,RochaReview:2023ilf}.
In particular, its retarded Green function, defined as the solution of the inhomogeneous equation
\(\partial_t G = \mathfrak{D}\partial_x^2 G + \delta(t)\delta(x)\) with \(G(0^-,x)=0\),
\vspace{-0.1cm}
\begin{equation}
G(t,x)=\Theta(t)\frac{e^{-x^2/(4\mathfrak{D}t)}}{\sqrt{4\pi \mathfrak{D}t}},
\end{equation}
exhibits Gaussian spatial tails extending to infinity at any \(t>0\), suggesting instantaneous propagation.
Furthermore, upon Lorentz transformation, equation \eqref{diffusion} develops exponentially growing Fourier modes at arbitrarily large wavenumber
\cite{Hiscock_Insatibility_first_order,Kost2000,GavassinoSuperlum2021},
rendering the boosted initial-value problem unstable and ill-posed.

These undesired features of diffusion have prompted two distinct responses. The prevailing approach has been to interpret them as genuine pathologies, and to eliminate them by modifying Fick’s law, giving rise to a broad class of hyperbolic relativistic dissipative theories \cite{Geroch_Lindblom_1991_causal,GavassinoUniveraalityI2023odx}, including Cattaneo-type models \cite{cattaneo1958,Jou_Extended}, Israel-Stewart theory \cite{Israel_Stewart_1979,Hishcock1983,Baier2008,Denicol2012Boltzmann}, divergence-type formulations \cite{Liu1986,GerochLindblom1990}, and the BDNK framework \cite{Bemfica2019_conformal1,Kovtun2019,BemficaDNDefinitivo2020,GavassinoAntonelli:2025umq}. An alternative viewpoint maintains that the diffusion equation is not intrinsically incompatible with relativity, but instead appears problematic only when extended beyond the class of solutions warranted by its microscopic origin \cite{Kost2000,Geroch1995,LindblomRelaxation1996,GerochCriticism:2001xs}. This perspective has recently been sharpened by proposals advocating ``approximate'' Lorentz transformations of diffusive dynamics, designed to exclude the unphysical branch of solutions \cite{ArmasDensityFrame2020mpr,Basar:2024qxd,GavassinoParabolic2025hwz}.

In this Letter, we resolve this long-standing tension by demonstrating that the diffusion equation is, indeed, fully compatible with relativity when interpreted within the correct microscopic framework. We construct a fully relativistic, causal, and stable kinetic theory whose hydrodynamic sector is governed \emph{exactly} by the diffusion equation, without introducing additional gradient corrections, wavenumber cutoffs, or a finite propagation speed. This explicit realization shows that, from a physical standpoint, the apparent acausality and instability of diffusion do not reflect a fundamental inconsistency: the commonly cited pathologies arise solely from applying the diffusion equation to non-normalizable configurations that lie outside the solution space permitted by the underlying relativistic dynamics, and relativistic causality (interpreted as the physical impossibility of superluminal messaging) is preserved once we keep track of all the available microscopic information. We thus revisit the standard arguments for acausality and instability and show that they fail once diffusion is treated as a consistent and complete sector of a relativistic kinetic theory.

{\it \noindent \textbf{Notation --}} Througout the article, we work in natural units, $c=\hbar=k_B=1$.

{\it \noindent \textbf{The Vlasov-Fokker-Planck equation --}} Consider a relativistic particle of mass \(m\,{>}\, 0\) confined, for simplicity, to one spatial dimension. Denoting its momentum by \(p\), its energy is \(\varepsilon\,{=}\,\sqrt{m^{2}{+}p^{2}}\), and its velocity is \(v\,{=}\,d\varepsilon/dp\,{=}\,p/\varepsilon\). We assume that the particle is immersed in a thermal medium at inverse temperature \(\beta\), and undergoes random momentum exchanges with it. Under these conditions, the particle experiences Brownian motion in momentum space. Accordingly, an ensemble of such particles is described by a kinetic distribution function \(f(t,x,p)\) satisfying the Vlasov-Fokker-Planck (VFP) equation
\cite{Debbasch,DunkelHanggi}
\vspace{-0.25cm}
\begin{equation}\label{vlasovfokkerplanck}
(\partial_t+v\,\partial_x)f
=\frac{1}{\beta^{2}\mathfrak{D}}\,
\partial_p
\left(
\partial_p f
+\beta v f
\right),
\end{equation}
\newpage\noindent where \((\beta^{2}\mathfrak{D})^{-1}\) denotes the momentum-space diffusivity (which is assumed to be a positive constant). Because the particle velocity obeys \(|v|\leq 1\), the theory is manifestly causal. It is moreover covariantly stable and thermodynamically consistent (see e.g. \cite{GavassinoDistrubingMoving:2026klp}). A quick way to verify this is to notice that the free-energy current
\vspace{-0.1cm}
\begin{equation}
\mathcal{F}^\mu = \frac{1}{\beta}\int \frac{dp}{2\pi}
\begin{bmatrix}
1\\
v
\end{bmatrix}
f(\beta\varepsilon+\ln f-1)
\end{equation}
has nonpositive divergence:
\vspace{-0.15cm}
\begin{equation}
\partial_\mu \mathcal{F}^\mu
= -\frac{1}{\beta^{3}\mathfrak{D}}
\int \frac{dp}{2\pi f}\,
\left(
\partial_p f+\beta v f
\right)^{2}
\leq 0\, ,
\end{equation}
in agreement with the second law of thermodynamics \cite{Groot1980RelativisticKT,huang_book,cercignani_book}. 
This, combined with the convexity of \(\mathcal{F}^0\), immediately establishes Lyapunov stability of the equilibrium state in any inertial frame \cite{GavassinoGibbs2021,GavassinoCausality2021}.

{\it \noindent \textbf{Main Theorem --}} We can now state our main result, which we express in the form of a theorem.
\vspace{-0.1cm}
\begin{theorem}\label{theo1}
Let \(g(k)\) be a Schwartz function satisfying \(g^{*}(k)=g(-k)\), and define (for $t\geq 0$)
\begin{equation}\label{magicsolutions}
f(t,x,p)
=
e^{\alpha-\beta\varepsilon}
+
\int_{\mathbb{R}} \frac{dk}{2\pi}\,
g(k)\,
e^{-\beta\varepsilon - i\beta \mathfrak{D} k p + i k x - \mathfrak{D} k^{2} t},
\end{equation}
with \(\alpha\in\mathbb{R}\).
Then, the following statements hold:
\begin{itemize}
  \setlength{\itemsep}{2pt}
  \setlength{\parskip}{0pt}
  \setlength{\parsep}{0pt}
\item[\textup{\bf(a)}]
\(f(t,x,p)\) is a smooth, real-valued function, well defined for all \(t\ge0\), and with finite energy-momentum moments.

\item[\textup{\bf(b)}]
\(f(t,x,p)\) is nonnegative for all \(t\ge0\), provided that
\(
\|g\|_{L^{1}}\le 2\pi e^{\alpha}.
\)

\item[\textup{\bf(c)}]
\(f(t,x,p)\) is an exact solution of the relativistic VFP equation~\eqref{vlasovfokkerplanck}.

\item[\textup{\bf(d)}]
The associated particle density
\(
n(t,x)=\int \frac{dp}{2\pi}\,f(t,x,p)
\)
satisfies the diffusion equation~\eqref{diffusion}, with initial condition
\vspace{-0.1cm}
\begin{equation}\label{densitykernelrelation}
n(0,x)
=
\frac{e^{\alpha} m K_{1}(m\beta)}{\pi}
+
\int_{\mathbb{R}} \frac{dk}{2\pi}\,
g(k)\,
\frac{m\,K_{1}\!\left(m\beta\sqrt{1+\mathfrak{D}^{2}k^{2}}\right)}
{\pi\sqrt{1+\mathfrak{D}^{2}k^{2}}}
\,e^{ikx}.
\end{equation}
\end{itemize}
\end{theorem}

\begin{proof}
\textup{\bf(a)}
For \(t\ge0\), the integrand in equation \eqref{magicsolutions} is a Schwartz function of \(k\), ensuring absolute convergence of the integral. Differentiation with respect to \(t\), \(x\), or \(p\) introduces at most polynomial factors in \(k\), which preserve convergence and smoothness. Reality of \(f\) follows directly from the condition \(g^{*}(k)=g(-k)\). The finiteness of all moments $\int \frac{dp}{2\pi} \varepsilon^a p^b f(t,x,p)$ (with $a,b\geq 0$) follows from the fact that $g(k)e^{-\beta \varepsilon}$ is a Schwartz function of $(k,p)$.

\noindent\textup{\bf(b)}
Using the triangle inequality, one finds
\vspace{-0.1cm}
\begin{equation}
\left|
\int_{\mathbb{R}} \frac{dk}{2\pi}\,
g(k)\,
e^{-\beta\varepsilon - i\beta \mathfrak{D} k p + i k x - \mathfrak{D} k^{2} t}
\right|
\le
\frac{e^{-\beta\varepsilon}}{2\pi}\,\|g\|_{L^{1}} .
\end{equation}
The stated bound on \(\|g\|_{L^{1}}\) therefore guarantees \(f(t,x,p)\ge0\).

\noindent\textup{\bf(c,d)} The last two statements follow by direct substitution. All derivatives may be interchanged with the \(k\)-integral, and the integrand is a Schwartz function also in \(p\), allowing the exchange of \(p\)- and \(k\)-integrations.
\end{proof}

The theorem establishes that any sufficiently regular solution of the diffusion
equation \(n(t,x)=n_{0}+\delta n(t,x)\) (defined for \(t\ge0\)), with \(n_{0}\) a constant background and \(\delta n\) a localized perturbation, admits an
\emph{exact} embedding into a corresponding VFP solution \(f(t,x,p)\) of the form~\eqref{magicsolutions}.
This embedding is realized by choosing
\begin{equation}\label{geometro}
g(k)=\delta n(k)\,
\frac{\pi\sqrt{1+\mathfrak{D}^{2}k^{2}}}
{m\,K_{1}\!\left(m\beta\sqrt{1+\mathfrak{D}^{2}k^{2}}\right)}\, ,
\end{equation}
where \(\delta n(k)\) denotes the Fourier transform of the initial data
\(\delta n(0,x)\). Provided that \(g(k)\) is sufficiently regular and the
perturbation \(\delta n\) sufficiently small, the resulting distribution
\(f(t,x,p)\) defines a physically admissible kinetic solution whose associated
particle density coincides identically with \(n(t,x)\). In this
sense, Fickian diffusion is an exact hydrodynamic sector of the
relativistic kinetic theory, rather than an approximation or a long-wavelength
limit.

We now clarify which initial data admit a well-defined embedding. If the initial perturbation
\(\delta n(0,x)\) belongs to the Schwartz class, then its Fourier transform \(\delta n(k)\) is also
Schwartz, implying that the corresponding \(g(k)\) is smooth. However, \(g(k)\) may not itself be
Schwartz, since the embedding kernel in \eqref{geometro} grows exponentially as
\(e^{m\beta|\mathfrak{D}k|}\) at large \(|k|\). Nevertheless, for any strictly positive time \(t>0\),
the integral representation~\eqref{magicsolutions} acquires the Gaussian factor
\(e^{-\mathfrak{D}k^{2}t}\), which dominates this exponential growth and ensures convergence. Thus,
the embedding is unambiguous for all \(t\,{>}\,0\) whenever \(\delta n(0,x)\) is Schwartz (while the
instant \(t\,{=}\,0\) may require interpretation in a distributional sense). A sufficient condition for the embedding to define a bona fide kinetic distribution already at
\(t\,{=}\,0\) is that
\(
\delta n(k)\,e^{(m\beta\mathfrak{D}+a) |k|}\in L^{1}(\mathbb{R})
\) (for some arbitrary $a>0$),
for instance if \(\delta n(0,x)\) extends analytically to the strip
\(|\mathfrak{Im}\, x| < m\beta \mathfrak{D}+a\).

Let us finally note that, although Theorem~\ref{theo1} assumes $g(k)$ to be Schwartz for simplicity, the
construction only relies on convergence of the $k$ integral for $t>0$, and therefore
extends to substantially less regular data, suggesting that a broader class of
initial profiles $\delta n(0,x)$ (e.g. in $L^p$ or $H^s$) might admit an embedding, in a distributional sense.

\newpage

{\it \noindent \textbf{Why stability is preserved --}}
The standard argument for the instability of the diffusion equation proceeds as follows \cite{Kost2000,GavassinoLyapunov_2020,GavassinoSuperlum2021}. One considers a density perturbation of the form
\(\delta n = e^{\Gamma \tilde t}\), where \(\tilde t=\gamma(t-Vx)\) is the time coordinate of an inertial observer moving with velocity \(V\) relative to the medium, and \(\gamma=(1-V^{2})^{-1/2}\). Substituting this ansatz into \eqref{diffusion} yields \(\Gamma=(\mathfrak{D}\gamma V^{2})^{-1}>0\), suggesting that a boosted observer would detect exponential growth of spatially homogeneous perturbations (signaling an instability \cite{Hiscock_Insatibility_first_order}).

We now show that such pathological solutions do not admit an embedding into relativistic kinetic theory. Indeed, configurations proportional to \(e^{\Gamma\gamma(t-Vx)}\) are not spatially localized and therefore fall outside the class of solutions covered by our theorem. More strongly, one can show that any solution of the relativistic VFP equation of the form
\(
f(t,x,p)=e^{\alpha-\beta\varepsilon}+\delta f(p)\,e^{\Gamma\gamma(t-Vx)}
\)
necessarily satisfies \(\Gamma\le0\), and is thus stable. To see this, note that for this spacetime dependence,  equation \eqref{vlasovfokkerplanck} reduces to
\vspace{-0.2cm}
\begin{equation}
\Gamma(1-vV)\,\delta f
=
\frac{1}{\beta^{2}\mathfrak{D}\gamma}\,
\partial_p
\bigl(
\partial_p\delta f+\beta v\,\delta f
\bigr).
\end{equation}
Multiplying both sides by \(e^{\beta\varepsilon}(\delta f)^{*}\), integrating over momentum, and solving for \(\Gamma\), one finds
\vspace{-0.1cm}
\begin{equation}
\Gamma
=
-\frac{\displaystyle
\int \frac{dp}{2\pi}\,e^{\beta\varepsilon}
\bigl|
\partial_p\delta f+\beta v\,\delta f
\bigr|^{2}}
{\displaystyle
\beta^{2}\mathfrak{D}\gamma
\int \frac{dp}{2\pi}(1-vV)\,e^{\beta\varepsilon}|\delta f|^{2}}
\le0,
\end{equation}
where we used the fact that \(1-vV\geq 1-|V|>0\), and we assumed that $\delta f$ decays to 0 at large $p$.

In summary, the unstable modes of the diffusion equation lie outside the space of solutions that can be embedded into a relativistic kinetic theory. Their existence therefore does not signal a failure of diffusion as a relativistic model. Rather, it reflects the necessity (peculiar to relativity) of restricting Fick’s law to spatially localized initial data. Within this physically admissible solution space, diffusion is again stable \cite{GavassinoLorentzBoostedDiffusion:2026fff}.\footnote{\label{footonafootona} In a companion paper \cite{GavassinoFokkerQuasi:2026zsz}, we compute the full quasi-normal-mode spectrum of the VFP equation. We find that in both the Newtonian and ultrarelativistic regimes the hydrodynamic mode has the diffusive form \(\omega=-i\mathfrak{D}k^{2}\), with no extra corrections. However, while in the Newtonian case this mode exists for all \(k\in\mathbb{C}\), in the ultrarelativistic case it is confined to the strip \(|\mathfrak{Im}\,k|<(2\mathfrak{D})^{-1}\), ensuring that the covariant stability bound \(\mathfrak{Im}\,\omega\le|\mathfrak{Im}\,k|\) \cite{HellerBounds2022ejw,GavassinoBounds2023myj} is respected.}

{\it \noindent \textbf{Why microscopic causality is preserved --}}
The claim that the diffusion equation is acausal stems from the
observation that an initially compactly supported density perturbation
\(\delta n(t,x)\) develops unbounded spatial support for \(t>0\), thereby suggesting that the propagation is instantaneous \cite{MorseFeshbach1953,rauch2012partial}. We now show that, when such solutions are
embedded into the relativistic VFP framework, their
evolution is, in fact, perfectly causal, from a physical standpoint.

All perturbations of the form~\eqref{magicsolutions} can be written as
\(\delta f(t,x,p)=e^{-\beta\varepsilon}\,F(t,x-\beta\mathfrak{D}p)\).
It follows that, if there exists a phase-space point \((x_{0},p_{0})\) such that
\(\delta f(0,x_{0},p_{0})\propto F(0,x_{0}-\beta\mathfrak{D}p_{0})\neq0\),
then for any other spatial location \(x_{1}\) one can find a momentum
\(p_{1}=p_{0}+(x_{1}-x_{0})/(\beta\mathfrak{D})\) for which
\(\delta f(0,x_{1},p_{1})\propto F(0,x_{0}-\beta\mathfrak{D}p_{0})\neq0\).
As a consequence, no kinetic distribution of the form~\eqref{magicsolutions}
can possess compact spatial support at the initial time\footnote{This implies that commonly assumed preparation protocols (such as generating a spatially localized perturbation on top of an otherwise equilibrium microscopic state) do not, in general, lie within the diffusive sector of the relativistic kinetic theory considered here.}. This observation clarifies the origin of the apparent superluminal behavior.
While the macroscopic density perturbation \(\delta n(t,x)\) may exhibit
instantaneous spatial tails, the initial phase-space perturbation
\(\delta f(0,x,p)\) is already nonzero outside the support of \(\delta n(0,x)\).
No information is therefore transmitted faster than light. The nonlocal tails
of \(\delta n\) arise from a redistribution between positive and negative
contributions of \(\delta f\), fully determined by the initial microscopic data
within the past light cone (see \cite{GavassinoDisperisons2023mad} for a similar argument).

In summary, diffusion appears acausal only if ``signals'' are defined in terms of the macroscopic density \(\delta n\) alone. Once the complete microscopic
information encoded in the kinetic distribution \(\delta f\) is taken into
account, the evolution is manifestly causal, and the apparent superluminal
tails of \(\delta n\) convey no new information to an observer far away.

{\it \noindent \textbf{The massless limit --}} It is instructive to consider a case in which the kinetic distribution can be
obtained explicitly. In the ultrarelativistic limit \(m\rightarrow 0\), the embedding
kernel in equation \eqref{geometro} reduces to the polynomial
\(\pi\beta(1+\mathfrak{D}^{2}k^{2})\), allowing the integral
representation~\eqref{magicsolutions} to be evaluated in closed form:
\vspace{-0.1cm}
\begin{equation}\label{fmassles}
f(t,x,p)
=
e^{\alpha-\beta|p|}
+\pi \beta\, e^{-\beta|p|}
\bigl(1-\mathfrak{D}^{2}\partial_{x}^{2}\bigr)\,
\delta n\!\left(t,x-\beta\mathfrak{D}p\right).
\end{equation}
One readily verifies that, for any solution \(\delta n(t,x)\) of the diffusion
equation (not necessarily localized), the distribution function
\eqref{fmassles} indeed satisfies the ultrarelativistic VFP equation. By contrast,
showing that the associated particle density
\(
n(t,x)=\int \frac{dp}{2\pi}\,f(t,x,p)
\)
coincides with \(n_{0}+\delta n(t,x)\) requires an integration by parts in
momentum space:
\vspace{-0.1cm}
\begin{equation}\label{proofofDensity}
\begin{split}
n(t,x)
={}&
n_{0}
+\int \frac{dp}{2}\,\beta\,e^{-\beta|p|}
\bigl(1-\beta^{-2}\partial_{p}^{2}\bigr)\,
\delta n\!\left(t,x-\beta\mathfrak{D}p\right)\\
={}&
n_{0}
+\int \frac{dp}{2}\,\beta\,
\bigl(1-\beta^{-2}\partial_{p}^{2}\bigr)e^{-\beta|p|}
\,\delta n\!\left(t,x-\beta\mathfrak{D}p\right)
=n_{0}
+\int dp\,\delta(p) 
\,\delta n\!\left(t,x-\beta\mathfrak{D}p\right)
=
n_{0}+\delta n(t,x).
\end{split}
\end{equation}
\newpage\noindent This step is justified only within the subset of diffusive solutions for which
\(\delta n\) does not grow too rapidly at spatial infinity. As a counterexample,
consider the unstable modes \(\delta n=e^{(t-Vx)/(\mathfrak{D}V^{2})}\) discussed
earlier. In this case, the integrand in \eqref{proofofDensity} behaves as
\(e^{-\beta|p|+\beta p/V}\), which diverges exponentially in one momentum
direction, resulting in an infinite particle density. These unstable modes
therefore lie outside the admissible solution space selected by the kinetic
theory.

The explicit form~\eqref{fmassles} also clarifies the mechanism by which causality
is preserved. Suppose that the initial density perturbation \(\delta n(0,x)\)
is supported within the interval \(-R\le x\le R\). For a kinetic distribution
of the form~\eqref{fmassles}, an observer located far away, at \(x=L\gg R\), can
nevertheless register a weak precursor in the momentum distribution at his/her location, appearing
as a localized ``bump'' confined to the range
\((L-R)/(\beta\mathfrak{D})\le p\le (L+R)/(\beta\mathfrak{D})\). This
momentum-space feature already encodes the complete spatial profile of
\(\delta n\), since the operator \(1-\mathfrak{D}^{2}\partial_{x}^{2}\) is
invertible. It follows that no new information is conveyed when the apparently
superluminal density signal subsequently reaches the observer: all relevant
information was already contained in the initial kinetic data.

{\it\noindent\textbf{Generalization to higher dimensions --}}
We emphasize that our restriction to one spatial dimension was adopted solely for
convenience.
The construction extends straightforwardly to $d$ spatial dimensions upon replacing
$kx$ and $kp$ by the inner products $\mathbf k\!\cdot\!\mathbf x$ and
$\mathbf k\!\cdot\!\mathbf p$, and promoting the operator
$\partial_p(\partial_p+\beta v)$ to $\nabla_{\mathbf p}\!\cdot(\nabla_{\mathbf p}+\beta\mathbf v)$.
With these substitutions, the proof of Theorem~\ref{theo1} is unchanged.
The only modification concerns the explicit form of the density kernel appearing in
equations \eqref{densitykernelrelation} and \eqref{geometro}, which acquires a
dimension-dependent structure through the corresponding $d$-dimensional momentum
integral.

{\it \noindent \textbf{Physical interpretation and broader consequences --}} The fact that the solutions \eqref{magicsolutions} solve the diffusion equation at arbitrarily large $k$ does \textit{not} imply that diffusion is the dominant dissipation mechanism of VFP in all regimes. In fact, when gradients are large, the states \eqref{magicsolutions} are not generic, as they are explicitly fine-tuned to remove all non-hydrodynamic contributions, which are shown in a companion paper to have ballistic character \cite{GavassinoFokkerQuasi:2026zsz}. Only in the hydrodynamic regime (i.e. when gradients are small compared to $1/\mathfrak{D}$) are generic solutions well approximated by \eqref{magicsolutions}, and diffusion becomes the dominant sector (and thus the appropriate effective description) of the system.

The main point of this work, however, is not that diffusion is the only relevant transport process in relativistic VFP, but that the diffusive mode $\omega = -i\mathfrak{D}k^2$ \textit{exists as an exact excitation} even at short wavelengths. In all other known relativistic theories, this mode ceases to exist beyond some critical value of $k \in \mathbb{R}$, either because it collides with other excitations \cite{Bajec:2024jez,Brants2025SavingCausality} or because its character changes qualitatively due to higher-order corrections. Indeed, it was conjectured in \cite{Krotscheck1978} that $\omega/k$ always remains finite at large $k$ in relativistic causal theories, a claim that has had repercussions across the relativistic hydrodynamics literature \cite{Pu2010,RomatschkeReview:2017ejr,Kovtun2019}. More recently, it was argued \cite{HellerBounds2022ejw,HellerHydrohedron2023jtd} that the dispersion relation $\omega = -i\mathfrak{D}k^2$ cannot arise from causal microscopic theories, based on the assumption that, if it exists for all real $k$, it should also exist (by analytic continuation) at complex $k$, leading to instabilities under Lorentz boosts. In VFP, however, the diffusive mode ceases to exist as a genuine excitation before entering unstable regions of the complex $k$-plane (see footnote \ref{footonafootona}). Since even a single consistent counterexample suffices to invalidate a general prohibition, diffusion is not intrinsically incompatible with relativity. It can arise as a (possibility exact) sector of a microscopic theory, within which all pathological solutions are physically forbidden.

{\it \noindent \textbf{Discussion and Conclusions --}} We provided a rigorous counterexample to the widely held belief that the diffusion
equation (and, more generally, parabolic theories of dissipation) cannot consistently arise in a fully relativistic context. We have
demonstrated that a fully relativistic, causal, and stable microscopic theory\footnote{It is worth emphasizing that, while the Vlasov-Fokker-Planck equation is a perfectly acceptable relativistic kinetic theory when regarded as a stand-alone equation, it is also a strictly Markovian dynamical model. It has been suggested that such exact Markovianity may itself be an idealization in a relativistic
setting, in the sense that deriving this equation from an underlying Langevin description
can require a degree of temporal or spatial non-locality
\cite{Petrosyan:2021lqi,Zaccone2023}. On the other hand, the Langevin equation itself is subject to intrinsic limitations in
relativistic contexts. Since a relativistic Hamiltonian description of interacting classical
particles cannot exist \cite{CurrieJordanSudarshan1963NoInteraction,CannonJordan1969NoInteraction}, locality arguments based on such equations are not fundamental, and may be intrinsically unreliable.
}
can give rise to the diffusion equation~\eqref{diffusion} as an \emph{exact}
hydrodynamic sector, without the introduction of ultraviolet cutoffs,
higher-order gradient corrections, or finite relaxation times of the Cattaneo
type. The standard claims of instability are revealed to be
artifacts of extending diffusion beyond its physically admissible solution
space. When the analysis is restricted to smooth, spatially localized configurations in the rest frame of the medium, no instabilities arise, and the principle of causality (i.e. the impossibility of superluminal information transfer) is preserved once the full microscopic information carried by the kinetic
distribution is taken into account\footnote{Of course, the diffusion equation still defines an acausal (in the sense of \citet[\S~10.1]{Wald}) initial-value problem when regarded as a stand-alone partial differential equation. In this sense, diffusion is still parabolic. Nevertheless, it becomes compatible with the broader principle of relativistic causality once it is embedded into a larger microscopic theory with additional degrees of freedom, which redefine the notion of a physical signal. A closely analogous situation arises for the relativistic Schr\"odinger equation
\(
i\partial_t \psi=\sqrt{m^2-\partial_x^2}\,\psi,
\)
which by itself does not define a causal initial-value problem, yet corresponds to the positive-frequency sector of the Klein--Gordon equation
\(
-\partial_t^2 \psi=(m^2-\partial_x^2)\psi,
\)
which is a perfectly causal relativistic field theory \cite{GavassinoDisperisons2023mad}.}. All results presented here are derived within a linear model with constant coefficients; extensions to nonlinear dynamics are not addressed in this work.

The above results prompt a reassessment of the role played by relativity in
dissipative hydrodynamics. Our findings do not undermine the extensive body of
work devoted to constructing hyperbolic relativistic theories of dissipation,
such as Cattaneo, Israel-Stewart, BDNK, or divergence-type formulations \cite{Geroch_Lindblom_1991_causal,GavassinoUniveraalityI2023odx,cattaneo1958,Israel_Stewart_1979,Hishcock1983,Baier2008,Denicol2012Boltzmann,Liu1986,GerochLindblom1990,Bemfica2019_conformal1,Kovtun2019,BemficaDNDefinitivo2020,GavassinoAntonelli:2025umq}. In
practice, when one considers boosted frames, rotating fluids, or systems
without a globally comoving frame, the exact kinetic description leads to
effective equations that are highly nonlocal and difficult to handle
numerically. In such situations, hyperbolic hydrodynamic models remain
indispensable as controlled approximations that ensure manifest Lorentz
covariance and well-posedness at the level of partial differential equations.
Moreover, for many physical systems, Israel-Stewart-type theories provide a
quantitatively superior description of transient dynamics beyond the diffusive
regime \cite{Jou_Extended,Denicol_Relaxation_2011,Grozdanov2019,BAGGIOLI20201,BulkGavassino,GavassinoFarFromBulk:2023xkt,AhnBaggioli:2025odk,GavassinoPlasma:2025tul}.

At the same time, our analysis challenges a deeply rooted intuition: that
relativistic effects must necessarily manifest themselves at the hydrodynamic
level through a finite signal speed \cite{cattaneo1958,Israel_Stewart_1979}. The present results show that this
expectation is too restrictive. Relativity constrains the microscopic dynamics
and the admissible space of states, but it does not require the emergent
hydrodynamic equations themselves to be hyperbolic. In this sense, Cattaneo’s equation is not intrinsically more ``relativistic'' than Fick’s law; it is just more robust as a standalone macroscopic theory, remaining well-behaved over a broader class of states and frames. In other words, the true imprint of relativity lies
not in the form of the macroscopic equations alone, but in the structure of the
microscopic theory and in the restrictions it imposes on physically realizable
hydrodynamic solutions.

More broadly, our work highlights that apparent violations of causality or
stability in effective theories often signal a mismatch between the equation
and its domain of applicability, rather than a fundamental inconsistency. By
identifying diffusion as an exact, relativistically admissible sector of
kinetic theory, we provide a concrete example in which parabolic dynamics
coexist fully consistently with relativistic causality. We expect that similar
considerations may apply to other dissipative processes, and that a careful
separation between equations of motion and admissible solution spaces will
prove essential in future developments of relativistic nonequilibrium physics.

\vspace{-0.2cm}
\section*{Acknowledgements}
\vspace{-0.2cm}

I thank M.M. Disconzi, J. Noronha, J.F. Paquet, and A. Zaccone for reading the manuscript and providing useful feedback.
This work is supported by a MERAC Foundation prize grant,  an Isaac Newton Trust Grant, and funding from the Cambridge Centre for Theoretical Cosmology.

\bibliography{Biblio}

\label{lastpage}
\end{document}